\documentclass{article}

\usepackage{arxiv}
\usepackage{url}            
\usepackage{booktabs}       

\usepackage{nicefrac}       
\usepackage{microtype}      
\usepackage{lipsum}
\usepackage{algorithm}
\usepackage{amssymb,amsfonts}
\usepackage{amsthm}
\usepackage{lipsum}

\usepackage{algpseudocode} 
\usepackage{multirow}
\usepackage{graphicx}
\usepackage{textcomp}
\usepackage{xcolor}
\usepackage{setspace}
\usepackage{amsmath}  
\usepackage{array}
\usepackage{cite} 
\usepackage{mathtools}
\usepackage{subfigure}
\usepackage{graphicx}
\allowdisplaybreaks[4]

\usepackage{verbatim}

\newtheorem{theorem}{Theorem}[section]
\newtheorem{assumption}{Assumption}[section]

\DeclareMathOperator*{\argmax}{arg\,max}

\title{Q-greedyUCB: a New Exploration Policy for Adaptive and Resource-efficient Scheduling}

\author{
  Yu~Zhao \\ 
  Department of Electrical and Electronic Engineering\\
  Hanyang University\\
  Ansan 15588, South Korea  \\
  \texttt{zhaoyu0112@hanyang.ac.kr} \\
   \And
 Joohyun~Lee \\
  Department of Electrical and Electronic Engineering\\
  Hanyang University\\
  Ansan 15588, South Korea \\
  \texttt{joohyunlee@hanyang.ac.kr} \\
   \And
  Wei~Chen \\
  Department of Electronic Engineering and Beijing National Research Center\\ for Information Science and Technology\\
  Tsinghua University\\
  Beijing 100084, China \\
  \texttt{wchen@tsinghua.edu.cn} \\
}

\begin{document}

\maketitle
\begin{abstract}
This paper proposes a learning algorithm to find a scheduling policy that achieves an optimal delay-power trade-off in communication systems. Reinforcement learning (RL) is used to minimize the expected latency for a given energy constraint where the environments such as traffic arrival rates or channel conditions can change over time. 
For this purpose, this problem is formulated as an infinite-horizon Markov Decision Process (MDP) with constraints.
To handle the constrained optimization problem, we adopt the Lagrangian relaxation technique to solve it. 
Then, we propose a variant of Q-learning, Q-greedyUCB that combines Q-learning for \emph{average} reward algorithm and Upper Confidence Bound (UCB) policy to solve this decision-making problem. We prove that the Q-greedyUCB algorithm is convergent through mathematical analysis. 
Simulation results show that Q-greedyUCB finds an optimal scheduling strategy, and is more efficient than Q-learning with the $\varepsilon$-greedy and Average-payoff RL algorithm in terms of the cumulative reward (i.e., the weighted sum of delay and energy) and the convergence speed. We also show that our algorithm can reduce the regret by up to 12\% compared to the Q-learning with the $\varepsilon$-greedy and Average-payoff RL algorithm.
\end{abstract}


\section{Introduction}
People rely on modern communication technologies in everyday lives, via real-time text messaging, audio and video calls, and video streaming services over the Internet. Therefore, for handheld communication devices, low latency, and long battery lifetime are essential requirements of mobile users. Recently, with the emergence of the fifth-generation (5G) of cellular networks, there will be a large number of communication base stations and billions of mobile terminals connected to each other, and therefore low energy consumption and low latency of communication systems become more urgent~\cite{alsharif2019energy,ni2019research}.   
For typical communication systems, under fixed channel conditions, power consumption is a convex function of the transmission rate~\cite{chen2017delay,chen2017delayicc,chen2007optimal}. This means that as the transmission rate increases, the delay decreases at the cost of increased power consumption per bit. Therefore, there exists a trade-off between delay and power consumption. In this work, we aim to characterize the trade-off between delay and power in communication systems and propose a reinforcement learning (RL)-based method to solve this problem. More specifically, we use the RL algorithm to obtain an optimal scheduling strategy for a given average power constraint.

Scheduling under delay or energy constraints has been studied in~\cite{chen2017delay,chen2017delayicc,chen2007optimal,wang2017delay,zhao2019DelayOptimalAE,liu2017delay,sharma2018accelerated,djonin2007q,bae2019beyond}. 
Among these works, the Linear Programming problem (LP) method is adopted to obtain the optimal trade-off between delay and power consumption in~\cite{chen2017delay,chen2017delayicc,wang2017delay,zhao2019DelayOptimalAE,liu2017delay}. Also, the Lagrange relaxation technique has been applied in~\cite{chen2017delay,chen2017delayicc,zhao2019DelayOptimalAE}, in order to transform a constrained optimization problem into an unconstrained problem.
To achieve the minimum delay for a given power constraint, Chen \textit{et al}. have modeled this problem as a constrained Markov decision process (CMDP) and then proposed an algorithm to efficiently obtain the optimal delay-power trade-off in~\cite{chen2017delay, chen2017delayicc}.
To achieve the minimum delay for a given power constraint in a communication system, Wang \textit{et al}. applied the Linear Programming (LP) to solve the optimization problem, and then a stochastic scheduling policy is used to address the delay-power trade-off in~\cite{wang2017delay}. 
Similarly, in order to minimize the delay under a power constraint in a communication system with Markov arrivals, Zhao \textit{et al}. proposed a threshold-based algorithm to obtain an optimal delay-power tradeoff in~\cite{zhao2019DelayOptimalAE}. 
In~\cite{liu2017delay}, to minimize the average latency, the authors proposed a stochastic scheduling policy under power constraints based on channel and buffer states.
In~\cite{sharma2018accelerated}, a novel accelerated RL algorithm is presented to solve scheduling problems in an online manner with a faster learning rate.

Existing studies on the optimization problem for delay and power mostly take conventional methods without a learning algorithm. A challenging problem that arises in this method is that the algorithm needs to be executed repeatedly as the environment changes over time. To this end, we propose a novel RL algorithm to obtain the optimal delay-power trade-off in this paper. Besides, we compare the performance of the proposed algorithm with other algorithms
that focus on long-term rewards to prove the efficient performance of our proposed algorithm.

In our previous work, we applied Q-learning for \emph{average} reward algorithm to solve the problem of queue scheduling in the communication system. The Q-learning algorithm~\cite{watkins1989learning} is a widely used model-free RL algorithm. Several researchers have proposed variants of the Q-learning algorithm to improve its performance. Such as speedy Q-learning~\cite{azar2011speedy}, Delayed Q-learnig~\cite{strehl2006pac}, HAQL ~\cite{Bianchi2004HeuristicallyAQ}. In particular, Jin \textit{et al}. proposed a Q-learning with the UCB (Upper Confidence Bound) exploration policy, and proved it achieves the optimal regret for finite-horizon MDP with discounted rewards~\cite{jin2018q}. Dong \textit{et al}. proposed a Q-learning with UCB algorithm for Infinite-horizon MDP with discounted rewards and proved it is sample efficient than the state-of-the-art RL algorithm~\cite{dong2019q}. However, most of the algorithms above focus on discount rewards. In other words, the agent's goal is to maximize the cumulative discounted reward, which is limited. Discounted RL methods cannot handle some infinite-horizon tasks because there are no terminal states and without discounting future rewards. As Mahadevan showed in~\cite{mahadevan1996average}, this method can lead to sub-optimal action. An alternative method is to maximize the \emph{average} reward. Unfortunately, the study of the average reward RL has received relatively little attention in the literature. The first average reward RL algorithm was proposed by Schwartz~\cite{schwartz1993reinforcement}.

To address the poor performance issues during the learning process (i.e., slow convergence speed and low average reward), this paper presents an RL algorithm named Q-greedyUCB for infinite-horizon MDP with \emph{average} reward. We formulate this problem as an MDP, and then the Lagrange relaxation technique is used to convert a constrained optimization problem into an unconstrained problem. Also, We mathematically prove the convergence of the proposed algorithm. 
The main contributions of this paper are as follows:
\begin{itemize}
	\item We propose a low delay scheduling algorithm Q-greedyUCB for a given power constraint in a single-queue single-server communication system. The proposed method is shown to achieve an optimal trade-off between delay and power. Also, it is mathematically proved that the proposed algorithm converges to an optimal scheduling policy. 
	\item For the Q-greedyUCB, we combine the Q-learning for the \emph{average} reward algorithm and UCB policy for more efficient learning. We show that the new  Q-greedyUCB algorithm can reduce the regret by up to 12\% compared to the traditional Q-learning with $\varepsilon$-greedy policy and Average-payoff RL algorithm in~\cite{singh1994reinforcement}.
\end{itemize}

The remainder of the paper is organized as follows. We describe the system model in Section~\ref{sec:model}, where the delay-power trade-off problem is formulated as an infinite-horizon MDP with constraints problem. In Section~\ref{sec:rl_model}, the RL methodology is presented in detail, including Upper Confidence Bound, Q-greedyUCB algorithm for \emph{average} reward.
The numerical simulation results will be discussed in Section~\ref{sec:simulation}. Concluding remarks and discussions of future research are provided in Section~\ref{sec:disscussion_conclusion}. To better illustrate, we list the important notations used in this paper in Table~\ref{table:1}, along with the descriptions.

\begin{table}
\caption{Notations and Definitions}
\label{table:1}
\setlength{\tabcolsep}{3pt}
\begin{tabular}{|p{70pt}<{\centering}|p{380pt}|}
\hline
Symbol& 
Definitions\\
\hline
$\mathcal{A} $& 
 The action space\\
$B$& 
Maximum buffer size\\
$M$& 
The number of data packets each data arrival \\
$C$& 
The maximum number of packets transmitted by the transmitter per time slot \\
$D_\pi$& 
The average latency \\
$E_\pi$& 
The average power consumption \\
$E_{th} $& 
The average power constraint \\
$\mathcal{S}$& 
 The state space\\
$P$& 
The transition probability matrix\\
$R $& 
The reward matrix\\
$q[t]$& 
The queue length at time slot $t$ \\
$c[t] $& 
The number of packets to be transmitted in time slot $t$ \\
$\tau[t]$& 
A binary variable that indicates whether there is new traffic arrival or not\\
$\alpha$& 
The parameter of the Bernoulli distribution\\
$d_{t}$& 
The queue delay in time slot $t$\\
$e_{t}$& 
The power consumption in time slot $t$\\
$\lambda$& 
The Lagrangian multiplier\\
$\delta$& 
The value that approximately an upper bound on the probability of the event\\
$\sigma$& 
The parameter that controls the degree of exploration\\
$\mu(s_t, a_t)$& 
The average reward of action $a_t$ at state $s_t$ up to time slot $t$\\
$\gamma_k$& 
The step size when the state-action pair $(s_t,a_t)$ is visited $k$ times\\
\hline
\end{tabular}
\label{tab1}
\end{table}

\section{System Model}
\label{sec:model}
In this paper, we consider a single-queue-single-server system with an adaptive transmitter, as shown in Fig.~\ref{fig:model}. The time is divided into discrete time slots, i.e., $t \in \{1,2,\cdots\}$. 
The traffic arrival is assumed to follow a Bernoulli distribution.
We define a binary variable $\tau[t]$ that indicates whether there is new traffic arrival or not (i.e., $\tau[t]=1$ if there is new traffic arrival).
The arrival probability $ \text{Pr}\{\tau[t] = 1\} = \alpha$ and $\text{Pr}\{\tau[t] = 0\} = 1 - \alpha$, where $\alpha$ denotes the parameter of the Bernoulli distribution.  

\begin{figure}[ht]
\centering
\vspace{-0.3cm}
\includegraphics[width=10cm]{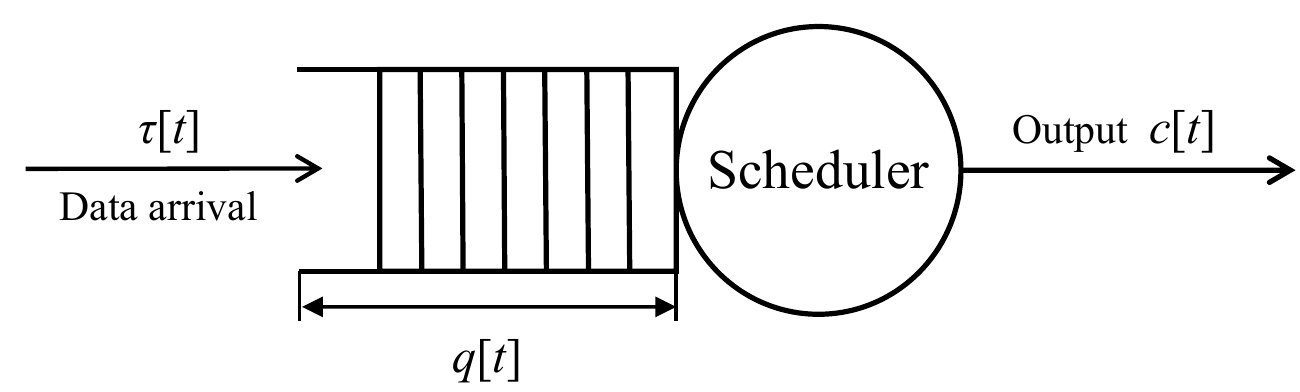}
\vspace{-0.3cm}
\caption{Our queue model with a single queue and a single server.} 
\label{fig:model}
\end{figure}

Arrived data packets are first added to a finite buffer with a maximum size of $B$. We assume that each data arrival contains $M$ data packets. The queue length is denoted by $q[t]$ at time slot $t$. The {\em Scheduler} determines the number of data packets to transmit based on the current queue length $q[t]$. Define $c[t]$ is the number of packets to be transmitted in time slot $t$. Due to the limitation of the transmitter, it can transmit up to $C$ data packets per time slot, and thus $c[t] \in \{0,1,\cdots,C\}$. Therefore, the queue length in the next time slot is given by
\begin{equation}\label{eq:next_state}
\begin{split}
q[t+1] = q[t] - c[t] + M\tau[t].\\
\end{split}
\end{equation}
To prevent buffer overflow or underflow, $c[t]$ needs to satisfy 
the following relationship:
\begin{equation}\label{eq:avoid}
\begin{split}
0\leq q[t] - c[t] \leq B - M.\\
\end{split}
\end{equation}
Moreover, the maximum buffer size and the number of packets arriving should satisfy $B>M$. 
Therefore, for a given queue length $q[t]$, 
$c[t]$ should satisfy
\begin{equation}
\begin{split}
\max(0,q[t]-B+M) \leq c[t] \leq \min(q[t],C).
\end{split}
\end{equation}
From Eq.~\eqref{eq:next_state}, the transition probability in the next time slot is given by

\begin{eqnarray} \label{eq:transition}
\begin{aligned}
&\text{Pr}\big\{q[t+1] = q'\ \vert \ c[t] = c,\ q[t] = q\big\}\\
&=\begin{cases}
    {\alpha  }         &      {\text{if}\; q' = q - c + M }, \\     
    {1 - \alpha}       &      {\text{if}\; q' = q - c  },\\
    {0}                    &  {\text{else.}}
\end{cases}
\end{aligned}
\end{eqnarray}
It shows that the future queue length depends only on the present queue length and the service rate. In other words, this process satisfies the Markov property, hence this system can be viewed as an MDP. 

In many communication systems, the latency decreases as the instantaneous service rate increases, while the power consumption increases.
Our goal is to achieve the minimum latency subject to the power constraint. Therefore, we have
\begin{equation}\label{eq:prob}
\begin{split}
{\bf P1: }& \min_\pi   D_\pi\\
& \text{subject to} \ E_\pi \leq E_{\text{th}} \\
\end{split}
\end{equation}
where $\pi$ represents a scheduling strategy, which describes a mapping from a state to probabilities of choosing available actions. $D_\pi$ and $E_\pi$ denote the average latency and energy consumption under strategy $\pi$, respectively. The average power constraint is denoted by $E_{\text{th}}$. 
We will show that the queue length corresponds to latency and it will be considered as a part of the reward function in Section~\ref{sec:rl_model}.

\section{Reinforcement learning methodology}
\label{sec:rl_model}
In this section, we first formulate the decision-making problem as an infinite-horizon MDP. Next, we describe the Q-greedyUCB algorithm and employ it to obtain an optimal scheduling policy. Then, we analyze the convergence of Q-greedyUCB algorithm. Finally, we discuss the properties of deterministic strategies and obtain optimal delay-power tradeoff curves.
\subsection{Reinforcement learning model}
\label{sec:rl_intro}
RL is widely used to solve optimization problems with model-unknown system (i.e., the state transition probability distribution is unknown). We model the RL framework on top of the system model presented in Section~\ref{sec:model}, as shown in Fig.~\ref{fig:rl_model}, where the {\em Scheduler} or transmitter is the agent, {\em Buffer} is the environment, the queue length $q[t]$ corresponds to the state $s_t$ and the number of packets to transmit in each time slot $t$, $c[t]$ is the action $a_t$.

In this paper, the scheduling problem is modeled as an infinite-horizon MDP. Our MDP model is defined by a 4-tuple, ($\mathcal{S}, \mathcal{A}, P, R$). $\mathcal{S}= \{0, 1, \cdots,{B} \}$ denotes the finite set of states (i.e., state space) and $s_t=q[t]\in\mathcal{S}$. $\mathcal{A}=  \{0, 1, \cdots, {C}\}$ represents the finite set of actions (i.e., action space), and $a_t=c[t] \in \mathcal{A}$. 
$P$ represents the transition probability matrix where $P_{ss'}^a$ is the probability of moving from current state $s$ to next state $s'$ under action $a$. $R$ is the reward matrix, where $R(s,a)$ denotes the immediate reward for the present state $s$  under action $a$.
The number of states and the number of actions are denoted by $S = |\mathcal{S}|$ and $A = |\mathcal{A}|$, respectively. 

\begin{figure}[t!]
\centering
\vspace{-0.3cm}
\includegraphics[width=10cm]{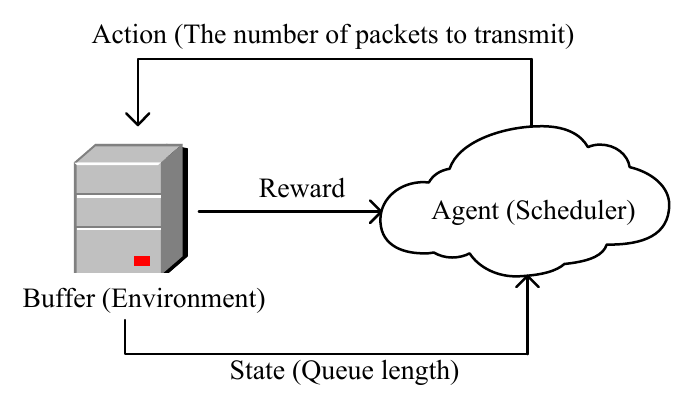}
\vspace{-0.3cm}
\caption{The RL structure of a communication system. }
\label{fig:rl_model}
\end{figure}

The problem {\bf P1} in Eq.~\eqref{eq:prob} is a constrained optimization problem. In order to obtain the optimal solution, we use the Lagrange multiplier method to solve it. By using Lagrangian relaxation technique, Eq.~\eqref{eq:prob} is transformed into the  following unconstrained problem.
\begin{equation}\label{eq:prob2}
\begin{split}
{\bf P2:} & \min_{\pi}  \ D_\pi + \lambda E_\pi  - \lambda E_{\text{th}}
\end{split}
\end{equation}
where $\lambda \geq 0$ is the Lagrangian multiplier.

Since the term $\lambda E_{\text{th}}$ has no effect on Eq.~\eqref{eq:prob2} under a given $E_{\text{th}}$ and specific $\lambda$, this term can be removed when we solve {\bf P2}. Then, we describe the immediate reward $R(s_t, a_t)$, as below:
\begin{equation}\label{eq:reward}
\begin{split}
& R(s_t, a_t) = - (d _t + \lambda e_t) \\
\end{split}
\end{equation}
where $d_t$ and $e_t$ are the delay and the energy consumption in time slot $t$, respectively. In Eq.~\eqref{eq:reward}, we can consider $\lambda$ as a tradeoff factor between latency and power consumption. Denote by $\overline{R}(s_t, a_t)$ the \emph{average} reward over the long run, and then the \emph{average} expected reward of policy $\pi$ is given by
\begin{equation}\label{eq:average_reward}
\begin{split}
\overline{R}^\pi(s_t, a_t) = \lim_{t \to \infty}E^\pi \{\frac{\sum_{t = 1}^\infty R(s_t, a_t)}{t}\}.  \\
\end{split}
\end{equation}

The Little's law states that the queue length is equal to the average delay times the data arrival rate~\cite{kulkarni2016modeling}. Therefore, the delay can be expressed as $d_t = q[t]/ (\alpha M)$ (i.e., the queue length is proportional to the average delay). For $e_t$, we assume a convex and increasing function of $c[t]$, such as $e_t ={ (c[t])}^2$.

\subsection{Q-greedyUCB algorithm for the average reward problem}
RL requires an exploration mechanism that feeds experience (samples) for diverse actions. The most common method is to use the $\varepsilon$-greedy policy that chooses an action randomly with some probability $\varepsilon$ and chooses a historically optimal action with probability $1-\varepsilon$~\cite{sutton2018reinforcement}. However, due to the randomness of the $\varepsilon$-greedy policy, the second-best action and the worst action can be chosen with the same probability. To avoid such inefficient exploration, the UCB method is widely used and proven to be asymptotically optimal to maximize the expected cumulative reward.
UCB is a type of multi-armed bandit algorithm based on optimism principles. The optimism principle refers to assigning a confidence bonus to each action based on the action-value observed so far. The higher the bonus value assigned to an action, the lower the agent's confidence in the action and the action will be selected more frequently in the future~\cite{lattimore2018bandit}. There are many variants of the UCB algorithm. Here, we use the most standard one, as below:
\begin{equation}\label{eq:UCB}
\begin{split}
a_t^{UCB} \leftarrow \argmax_{a_t\in \mathcal{A}}\left[ \mu(s_t, a_t) + \sigma\sqrt{\dfrac{\iota(k)}{k}} \right]
\end{split}
\end{equation}
where $\mu(s_t, a_t)$ denotes the average reward of action $a_t$ at state $s_t$
up to time slot $t$. 
$k$ indicates the number of times that $(s_t, a_t)$ has been selected up to time $t$.
We define the square-root term $\sigma\sqrt{\iota(k)/k}$, which is the confidence bonus that measures the uncertainty of the empirical mean of the current state-action pair $(s_t, a_t)$ in time slot $t$. In other words, this term measures how much the true mean of this action can be larger than the empirical mean. We will use $b_k = \sigma\sqrt{\iota(k)/k}$ for this term.
We will also use $b_k(s_t, a_t)$ to represent the confidence bonus term of the state-action pair $(s_t, a_t)$. 
The parameter $\sigma>0$ controls the degree of exploration. $\iota(k) = \ln(SAkt/\delta)$ is a log factor. The value $\delta$ is approximately an upper bound on the probability of the event. If $\delta$ is chosen to be very small, then the algorithm will exploit more (i.e., choose the best action in the history) and if $\delta$ is large, then the algorithm will explore more frequently.
Next, we formally describe the Q-greedyUCB algorithm.

In order to combine the Q-learning algorithm with the UCB algorithm to solve our problem, we modify the Q-learning with UCB algorithm with discounted rewards proposed by Dong \textit{et al}. in~\cite{dong2019q}. We adopt the Q-learning for the average reward problem algorithm proposed in~\cite{bertsekas1996neuro} and combine it with 
the UCB exploration strategy described in Eq.~\eqref{eq:UCB}.
Based on the above considerations, we propose the Q-greedyUCB algorithm, which is shown as the procedural form in Algorithm~\ref{alg:1}. 

\begin{algorithm}[ht]
\caption{Q-greedyUCB}\label{alg:1}
\begin{algorithmic}[1]
        \State Parameters: $\delta$, $\sigma$, $\varepsilon$.                                                                
		\State $Q(s, a), \hat{Q}(s, a), N(s, a)\leftarrow0,\forall s \in \mathcal{S}, a \in \mathcal{A}(s)$.
        \For{$t=1,2\cdots$}
                \State Choose an action $a_t\leftarrow \arg\max_{a'}\hat{Q}(s_t, a')$.
        \State Observe reward $R(s_t,a_t)$ and move to state $s_{t+1}$.
        \State $N(s_t, a_t)\leftarrow N(s_t, a_t) + 1$.
        \State $k\leftarrow N(s_t, a_t), \iota(k)\leftarrow \ln{\frac{SAkt}{\delta}}, b_k\leftarrow \sigma\sqrt{\dfrac{\iota(k)}{k}}$.
        \State Update $Q(s_t, a_t)$ according to Eq.~\eqref{eq:3}.
        \State $\hat{Q}(s_t, a_t)\leftarrow \min \big( \hat{Q}(s_t, a_t), Q(s_t, a_t)\big)$.
        \State $s_t \leftarrow s_{t+1}$.         
        \EndFor
\end{algorithmic}
\end{algorithm}

Q-greedyUCB uses the temporal difference method to update the action-value function. The update formula of Q-greedyUCB for the Q-function is given by
\begin{equation}\label{eq:3}
\begin{split}
Q(s_t, a_t) \leftarrow  (1 - \gamma_k)Q(s_t, a_t) + \gamma_k\Big( R(s_t, a_t) + \! &\max \limits_{a'\in \mathcal{A}{(s_{t+1})}} \!Q\!\left(s_{t+1}, a'\right) \!-\!  \max \limits_{v\in \mathcal{A}{(i)}} \!Q\!\left(i, v\right) \!+\! b_k\Big)
\end{split}
\end{equation}
where $s_{t+1}$ is a new state after action $a_t$ is executed in the present state $s_t$, and $a'$ is the action to be selected in $s_{t+1}$. 
Note that the confidence bonus term $b_k$ is added to the update term from the standard Q-learning.
We remind that $k = N(s_t, a_t)$ represents the number of times that the state-action pair $(s_t, a_t)$ is experienced up to time slot $t$.
${\gamma_k \in [0, 1]}$ represents the step size when the state-action pair $(s_t,a_t)$ is visited $k$ times. 
$\gamma_k=0$ indicates that the agent learn nothing, while when $\gamma_k = 1$, the agent only considers the current estimate. In RL the choice of the step size is crucial. In general, we set the step size equal to some constant value. However, it is not easy to obtain an optimal strategy when the problem is stationary~\cite{bertsekas1996neuro,sutton2018reinforcement}. Therefore, a feasible method is to set the step size to decrease over time. We set $\gamma_k = \phi/(k+\theta)$ in the proposed algorithm. In the standard Q-learning algorithm, the step size depends on current time $t$ (e.g., $\gamma_k = 1/t$). However, for the Q-greedyUCB algorithm, we need to consider how many  times each state-action pair is selected up to time $t$. Therefore, we set the step size as $\gamma_k = \phi/(k+\theta)$ for some positive constants $\phi$ and $\theta$. It means that the step size of the update is selected based on the number of times the state-action pair is selected. In Section~\ref{sec:convergence}, we will  discuss the required properties of the step size to guarantee convergence. The term $\max_{v\in {A(i)} } Q\left(i, v\right)$ denotes the optimal average reward expected to converge in each iteration, where $i$ represents the reference state. In line 9, $\hat{Q}(s,a)$ represents the historical minimum value of the Q-function. 

Our Q-greedyUCB algorithm decides the best action to play in a certain state by estimating the value of the state-action pair $Q(s, a)$ (i.e., Q-value). Each Q-value is stored in a matrix with the size $S\times A$, which we  call the Q-table. Before the algorithm starts learning, we need to set an arbitrary initial value for each $Q(s,a)$. In this paper, we set all $Q(s, a)=0$ initially. At the beginning of each time slot, the agent chooses an action 
based on UCB exploration policy, which corresponds to line 4 of the Algorithm~\ref{alg:1}. An agent that executes an action in a state will receive a reward and enters a new state, and then Q-table is updated. Each Q-value will not change after an adequate number of iterations. In other words, all Q-value in the Q-table will converge to the optimal value. Let $\pi^*$ denote an optimal scheduling policy. 
We can get an optimal strategy according to optimal Q-values, which is given by
\begin{equation}\label{eq:2}
\begin{split}
\displaystyle
\pi^*(\hat{a}|s)=\frac{1}{|\mathcal{A}^*(s)|} 
\text{ if }
\hat{a} \in \mathcal{A}^*(s) \\
\end{split}
\end{equation}

where $\pi^*(\hat{a}|s)$ denotes the probability to execute action $\hat{a}$ in state $s$, 
$\mathcal{A}^*(s) = \argmax_{a\in\mathcal{A}(s)}Q^*\left(s, a\right )$ is the set of optimal actions in state $s$,
and $Q^*$ refers to the optimal action-value function. 
Otherwise, $\pi^*({a}|s)=0$ if $a \notin \mathcal{A}^*(s)$ for non-optimal actions. 

\subsection{Analysis of convergence}\label{sec:convergence}
Now, we discuss the convergence of the Q-greedyUCB algorithm. In~\cite{abounadi2002stochastic}, for the average reward problem, the author has proposed a general framework for proving convergence based on the ODE (Ordinary Differential Equations) method. We first introduce two crucial assumptions from~\cite{abounadi2002stochastic}.  
\begin{assumption} \label{asump:1}
The step size $\gamma_k$ satisfies the following:
\begin{enumerate}
    \item ${\gamma_k}$ is an ideal tapering step size,
        i.e.,
        \begin{enumerate}
            \item $\sum_{k = 0}^\infty \gamma_k = \infty$.
            \item $\gamma_{k+1} \leq \gamma_k$ for sufficiently large $k$.
            \item There exists $g\in (0,1)$ such that 
                \begin{equation}
                    \sum_{k=0}^\infty\gamma_k^{1+\eta}<\infty, \quad \text{for} \quad \eta \geq g. \nonumber
                \end{equation}
            \item Let $\Gamma_k=\sum_{m=0}^k\gamma_m$. Then, for all $g \in (0,1),$
                \begin{equation}\label{eq:10}
                    \begin{aligned}
                    &\sup_k\frac{\gamma_{\lfloor gk \rfloor}}{\gamma_k}< \infty, 
                    \\
                    \end{aligned}
                \end{equation}
                \begin{equation}\label{eq:11}
                    \begin{aligned}
                    \lim_{k \to \infty}\frac{\Gamma_{\lfloor Gk \rfloor}}{\Gamma_k} = 1,  \text{uniformly in $G$} \in [g,1]
                    \end{aligned}
                \end{equation}
                where $\lfloor x \rfloor$ is the greatest integer less than or equal to $x$ (i.e., the floor function).
        \end{enumerate}
    \item There exists $\Delta>0$, for each state-action pair $(s, a)$ to be updated infinitely often such that
        \begin{equation}
            \begin{aligned}
            \liminf\limits_{t \to \infty}\frac{N(s_t,a_t)}{t+1} \ge \Delta \quad
            \text{with probability } 1, \forall s \in \mathcal{S}, a \in \mathcal{A}(s).
            \end{aligned}
        \end{equation}
\end{enumerate}
\end{assumption}

\begin{assumption}\label{asump:2}
A scalar real-valued function $f$ has the following properties:
\begin{enumerate}
    \item $f$ is Lipschitz continuous, i.e., $\vert f(x)-f(y)\vert\leq L_f\Vert{x - y}\Vert_\infty$ for some $L_f \in \mathbb{R}$.
    \item $f(x + ru) = f(x) + rf(u)$,\quad for $r \in  \mathbb{R}$.
    \item $f(u) < 0$.\footnote{In~\cite{abounadi2002stochastic}, this term is positive, $f(u) > 0$, because the default reward is a positive value. Since we set the reward as a negative value, $f(u) < 0$.}
\end{enumerate}
where $u=(1,\cdots,1)$ is a vector whose entries are all 1. 
\end{assumption}
Based on the above assumptions, we have the following.
\begin{theorem}\label{the:1}
If $\gamma_k = \phi/(k + \theta)$ for some positive constants $\phi$ and $\theta$,
by Algorithm~\ref{alg:1}, (1) Assumption~\ref{asump:1} holds, (2) the sequence $Q(s_t, a_t)$ is bounded, and (3) $Q$ converges to $Q^*$.
\end{theorem}

\begin{proof}[Proof]
We first prove that the step size satisfies Property 1 in Assumption~\ref{asump:1}. Assume $\gamma_k = \phi/(k + \theta)$ for some positive constants $\phi$ and $\theta$.
Then for property (a), we have
\begin{equation}
\begin{aligned}
&\sum_{k=0}^\infty\gamma_k = \int_{0}^\infty \frac{\phi}{k + \theta}\,dk = \phi\ln(k+\theta)|_0^\infty = \infty.\\
\end{aligned}
\end{equation}
It is easy to check the property (b) holds. For property (c), we have,
\begin{equation}
\begin{aligned}
\sum_{k=0}^\infty\gamma_k^{1+\eta}  = \int_{0}^\infty \left(\frac{\phi}{k + \theta}\right)^{(1+\eta)}\,dk = -\frac{\phi^{1+\eta}}{\eta(k+\theta)^\eta}|_0^\infty = \frac{\phi^{1+\eta}}{\eta\theta^\eta} < \infty.
\end{aligned}
\end{equation}
For property (d), we first consider the Eq.~\eqref{eq:10},
\begin{equation}
\begin{aligned}
&\frac{\gamma_{(\lfloor gk \rfloor)}}{\gamma_k} = \frac{\frac{\phi}{\lfloor gk \rfloor + \theta}}{\frac{\phi}{k + \theta}} = \frac{k+\theta}{\lfloor gk \rfloor + \theta} \leq \frac{1}{g} < \infty.\\
\end{aligned}
\end{equation}
The left-hand side of Eq.~\eqref{eq:11} is equivalent to
\begin{equation}
\begin{aligned}
&\lim_{k \to \infty}\frac{\Gamma_{\lfloor Gk \rfloor}}{\Gamma_k} = \frac{\int_{0}^{Gk}\gamma_k\,dk}{\int_{0}^{k}\gamma_k\,dk}
\end{aligned}
\end{equation}
by L'Hôpital's rule, reduces to 
\begin{align}\label{eq:16} \nonumber
\lim_{k \to \infty}\frac{\int_{0}^{Gk}\gamma_k\,dk}{\int_{0}^{k}\gamma_k\,dk} =\lim_{k\to\infty} \frac{G \gamma_{Gk}}{\gamma_k} = \lim_{k\to\infty} \frac{\frac{G\phi}{Gk + \theta}}{\frac{\phi}{k+\theta}} = \lim_{k\to\infty} \frac{G(k+\theta)}{Gk+\theta} = \lim_{k\to\infty} \frac{G + \frac{G\theta}{k}}{G + \frac{\theta}{k}} = 1.  
\end{align}
We complete the proof of Property 1. 

Next, we prove that Property 2 in Assumption~\ref{asump:1} is not required in this paper. Generally, to satisfy Property 2, $\varepsilon$-greedy policy is usually used to select actions. According to the law of large numbers, it can guarantee that each state-action pair is selected infinite times as $t \to \infty$. However, UCB policy is adopted in this paper. If we fix a state $s$, the UCB policy always selects the action with the highest reward value for each iteration under this state. From Eq~\eqref{eq:UCB}, we can see that if an action has been chosen only few times, the bonus term will be large. As a result, the confidence in this action will be low and thus leading to more frequent exploration. When there is high enough confidence in all the actions, UCB policy will always choose the action with the highest reward value without any additional exploration.
Moreover, if we draw a directed graph with the state space $\mathcal{S}$ and an edge from $s \in \mathcal{S}$ to $s' \in \mathcal{S}$ whenever the transition probability from $s$ to $s'$, $P_{ss'} > 0$ almost surely. Therefore, the graph is irreducible, i.e., there exists a path from any state $s \in \mathcal{S}$ to any state $s' \in \mathcal{S}$. We can draw the conclusion that the UCB policy will eventually find the optimal policy in each state, without requiring each state-action pair to be selected infinitely as $t \to \infty$.
This makes the convergence condition of Algorithm~\ref{alg:1} more relaxed. The proof of Assumption~\ref{asump:1} is complete. 

Under the Assumption~\ref{asump:1}, we can use the analysis of Borkar~\cite{borkar1998asynchronous} to establish the relationship between Eq.~\eqref{eq:3} and ODE, as below:
\begin{equation}\label{eq:ODE}
\begin{split}
\vec{Q}^{'} = H\vec{Q} - \vec{Q} - f(\vec{Q})u,    
\end{split}
\end{equation}
where $\vec{Q}^{'}$ denotes the derivative of vector $\vec{Q}$ and vector $u$ has the same length as vector $\vec{Q}$. $\vec{Q}$ is the vector in which each of the elements represents Q-value for each of state-action pair, i.e., $\vec{Q} = [Q(1,1),\cdots,Q(1, A),Q(2,1), \cdots,Q(2,A),\cdots,\\Q(S,1), \cdots,Q(S,A)]$. In~\cite{abounadi2002stochastic,bertsekas1996neuro}, it has been shown that $f(\vec{Q}) = \max_{v\in \mathcal{A}{(i)}} Q\left(i, v\right)$ satisfies the Assumption~\ref{asump:2}. $H$ is the mapping defined by
\begin{equation}\label{eq:mapping}
\begin{split}
H\vec{Q} = \vec{P}\left(\vec{R} + \vec{V} + \vec{b}\right),
\end{split}
\end{equation}
where vectors $\vec{R}$ and $\vec{b}$ refer to the immediate reward and confidence bonus for all state-action pairs, respectively. $\vec{P}$ is the transition probability matrix and $\vec{V} = \max_{a'\in \mathcal{A}{(s')}}Q\left(s', a'\right)u$. For the sake of intuition and proof, we express Eq.~\eqref{eq:mapping} as a non-vector form, as below:
\begin{equation}\label{eq:mapping_nonvector}
\begin{split}
(HQ)(s, a) = \sum_{s'\in \mathcal{S}} P_{ss'}^a\Big(\!R(s, a) + \!\max \limits_{a'\in \mathcal{A}{(s')}} \!Q\left(s', a'\right) + b_k(s, a)\!\Big).
\end{split}
\end{equation}

The convergence proof given in~\cite{abounadi2002stochastic} indicates that the mapping of $H$ needs to satisfy the following properties.
\begin{enumerate}
    \item $H$ is non-expansive with respect to the sup-norm:
        \begin{equation}\label{eq:property1}
             \Vert{HQ_1-HQ_2}\Vert_\infty\leq \Vert{Q_1 - Q_2}\Vert_\infty, \quad \forall Q_1, Q_2,
        \end{equation}
        where $\Vert{Q}\Vert_\infty = \max_{(s, a)}\vert Q(s, a)\vert$.
    \item $H$ satisfies:
        \begin{equation}\label{eq:property2}
            H(\vec{Q} + ru) = H\vec{Q} + ru, \quad \forall r \in \mathbb{R}.
        \end{equation}
\end{enumerate}

Now we prove that Eq.~\eqref{eq:mapping_nonvector} satisfies Property 1 (Eq.~\eqref{eq:property1}) for all $Q_1, Q_2$. For the sake of analysis, we consider the confidence bonus term is constant (In fact, this term decreases monotonically, and its value decreases to 0 as $t \to \infty$), where we define $(s^*, a^*) = \argmax_{(s,a)}|Q_1(s,a) - Q_2(s,a)|$ and use it to represent the state-action pair that has the maximum gap between $Q_1(s, a)$ and $Q_2(s, a)$. 
Then for all $(s,a)$,
\begin{align*}
\Vert{HQ_1 - HQ_2}\Vert_\infty &= \max_{(s, a)}\vert \sum_{s'\in \mathcal{S}}P_{ss'}^a[R(s, a) + \max \limits_{a'\in \mathcal{A}{(s')}} Q_1\left(s', a'\right) + b_k(s, a) \\
&- R(s, a) - \max \limits_{a'\in \mathcal{A}{(s')}} Q_2\left(s', a'\right) - b_k(s, a)]\vert\\
& = \max_{(s, a)}\vert \sum_{s'\in \mathcal{S}}P_{ss'}^a[\max \limits_{a'\in \mathcal{A}{(s')}} Q_1\left(s', a'\right) - \max \limits_{a'\in \mathcal{A}{(s')}} Q_2\left(s', a'\right)]\vert\\
& \leq \max_{(s, a)}\sum_{s'\in \mathcal{S}}\vert P_{ss'}^a[\max \limits_{a'\in \mathcal{A}{(s')}} Q_1\left(s', a'\right) - \max \limits_{a'\in \mathcal{A}{(s')}} Q_2\left(s', a'\right)]\vert\\
& = \max_{(s, a)}\sum_{s'\in \mathcal{S}}P_{ss'}^a\vert \max \limits_{a'\in \mathcal{A}{(s')}} Q_1\left(s', a'\right) - \max \limits_{a'\in \mathcal{A}{(s')}} Q_2\left(s', a'\right)\vert\\
& \leq \max_{(s, a)}\sum_{s'\in \mathcal{S}}P_{ss'}^a\vert Q_1\left(s^*, a^*\right) - Q_2\left(s^*, a^*\right)\vert\\
& = \vert Q_1\left(s^*, a^*\right) - Q_2\left(s^*, a^*\right)\vert\\
& = \max_{(s, a)}\vert Q_1\left(s, a\right) - Q_2\left(s, a\right)\vert\\
& = \Vert{Q_1 - Q_2}\Vert_\infty.
\end{align*}
For Property 2 (Eq.~\eqref{eq:property2}), for all $(s,a)$, we have,
\begin{align*}
(H(\vec{Q} + ru))(s,a) &= \sum_{s'\in \mathcal{S}} P_{ss'}^a\left(R(s, a) + \max \limits_{a'\in \mathcal{A}{(s')}} Q\left(s', a'\right) + b_k(s, a) + r\right)\\
& = \sum_{s'\in \mathcal{S}} P_{ss'}^a\left(R(s, a) + \max \limits_{a'\in \mathcal{A}{(s')}} Q\left(s', a'\right) + b_k(s, a) \right) + \sum_{s'\in \mathcal{S}} P_{ss'}^a r\\
& = \sum_{s'\in \mathcal{S}} P_{ss'}^a\left(R(s, a) + \max \limits_{a'\in \mathcal{A}{(s')}} Q\left(s', a'\right) + b_k(s, a) \right) + r\\
& = (HQ)(s,a) + r.
\end{align*}
Thus, we can conclude that the mapping $H$ satisfies the Properties 1 and 2 (i.e., Eqs.~\eqref{eq:property1} and~\eqref{eq:property2}). Besides, Assumptions~\ref{asump:1} and~\ref{asump:2} hold for the step size $\gamma_k$ and function $f$, respectively.  
Now, we can use the convergence proof structure proposed by Abounadi in~\cite{abounadi2002stochastic} (i.e., the convergence theorem for the asynchronous algorithm). Under the above analysis, we can conclude that the ODE (Eq.~\eqref{eq:ODE}) has a unique globally stable equilibrium $Q^*$. It means that $Q$ converges to $Q^*$, the proof is complete.
\end{proof}

\subsection{optimal delay-power trade-off curve}
In this subsection, we obtain the feasible regions of average delay-power points for all deterministic policies. Furthermore, we introduce a significant property for all deterministic policies (i.e., all deterministic policies are threshold-based). 

In~\cite{chen2017delay}, the authors reveal that the optimal scheduling policies are threshold-based. In other words, the number of packets transmitted per time slot is proportional to the queue length. Moreover, the feasible average latency and power region is a convex polygon. 
Define $\hat{q} \ \text{and} \ q \in \{0,1,\cdots,B\}$, and then we have 
\begin{equation}\label{eq:threshold-based}
\begin{split}
c_{\pi}(q) \leq c_{\pi}(\hat{q}) \quad \text{if and only if} \ q \leq \hat{q},
\end{split}
\end{equation}
where $c_{\pi}(q)$ and $c_{\pi}\hat{q}\in \{0,1,\cdots,C\}$ refer to the number of packets transmitted for a deterministic policy $\pi$ when the queue length $q[t] = q$ and $q[t] = \hat{q}$, respectively. 
According to Formula~\eqref{eq:threshold-based}, we can obtain all deterministic strategies that satisfy the threshold-based property.
\begin{figure}[ht]
  \centering
    \includegraphics[width=10cm]{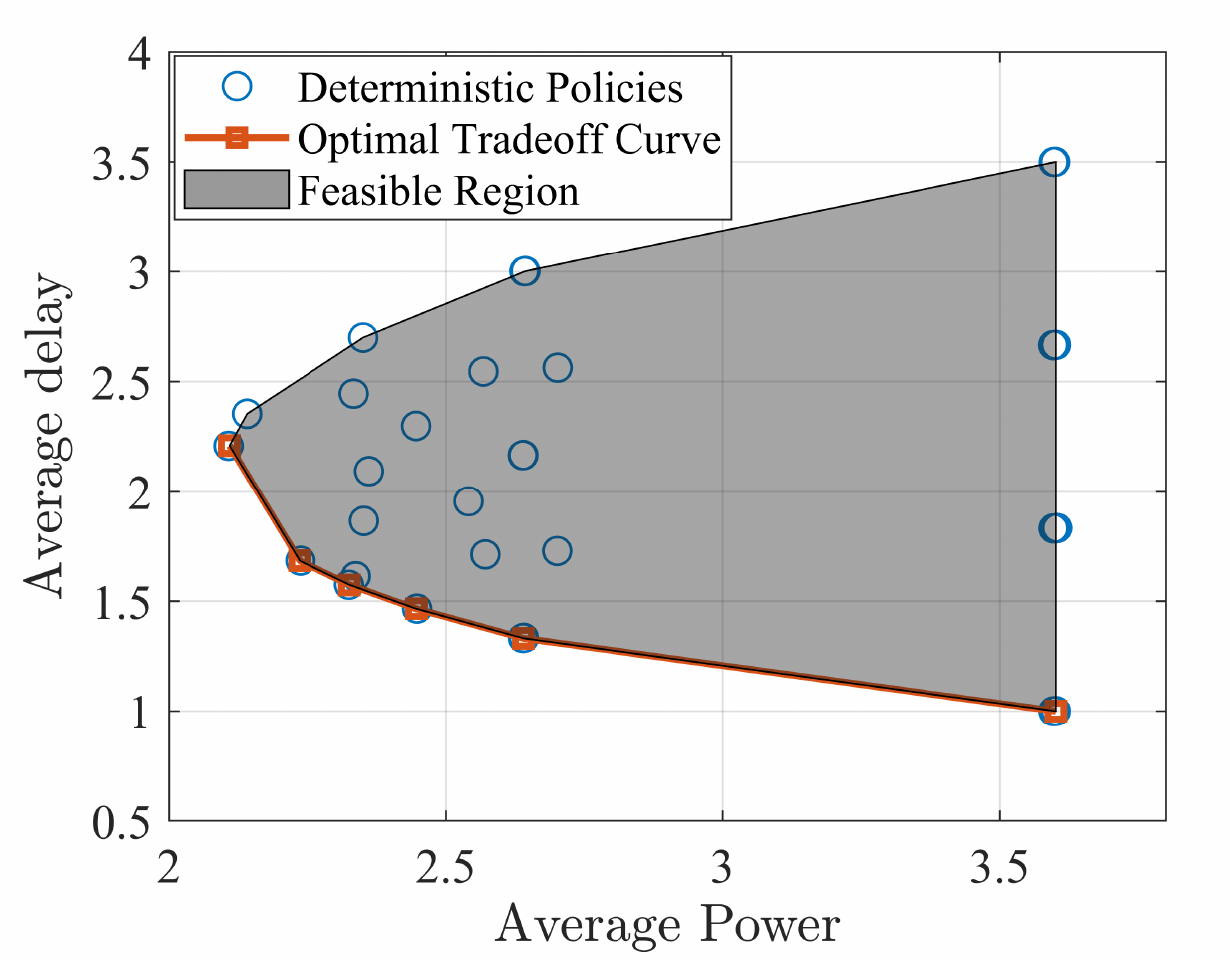}
    \caption{Points corresponding to all deterministic strategies and optimal delay-power tradeoff curves. $B = 6$, $M = 3$, $C = 3$, $\sigma = 1$, $\varepsilon = 0.01$, $\delta = 0.01$, and $\lambda = 1$).}
    \label{fig:delay power plane}
\end{figure}

Fig.~\ref{fig:delay power plane} depicts the delay-power pairs of all deterministic strategies that satisfy threshold-based property, which are shown in "o" markers (To simplify the figure, we set $B = 6$ and $M = C = 3$). From the Corollary given in~\cite[Corollary 3]{chen2017delay}, we have that the optimal delay-power tradeoff curve is decreasing and piecewise linear, is shown in Fig.~\ref{fig:delay power plane} by the red line. We solve the unconstrained problem {\bf P2} with a specific $\lambda$ by using the Algorithm~\ref{alg:1}. Therefore, the optimal scheduling policy is on the optimal delay-power tradeoff curve.

\section{Simulation Result}
\label{sec:simulation}
In this section, we present the simulation result. To evaluate the performance of the Q-greedyUCB algorithm, we implement a MATLAB simulation with different input parameters. There are two cases:~\romannumeral1) change the maximum buffer size $B$, the number of packets in each data arrival $M$, and the maximum number of transmitted packets in each time slot $C$ under the condition of constant arrival rate $\alpha$.~\romannumeral2) change the parameters $\alpha$ under the condition of constant $B$, $M$, and $C$. Additionally, we compare the Q-greedyUCB algorithm with Policy Iteration (PI) in~\cite{chen2017delay}, the Q-learning algorithm in~\cite{zhao2019reinforcement} and the Average-payoff RL algorithm (ARL) in~\cite{singh1994reinforcement}.

For the simulation scenario, we consider the case where $C = 4$ and $C = 5$. When we set $C = 4$, that means we can adopt four optional modulations to transmit 1, 2, 3, or 4 packets in a timeslot, respectively.
We set $\gamma_k = 1/(k+1)$, $\alpha = 0.4$,  $\sigma = 1$, $\varepsilon = 0.01$, $\delta = 0.01$ and $\lambda = 1$. We use $e_t ={ (c[t])}^2$ for the power consumption.
The optimal scheduling strategies generated by policy iteration (PI), ARL, Q-learning, and Q-greedyUCB algorithms are depicted in Fig.~\ref{fig:policy} under different parameters of $B$, $M$, and $C$.
We can see that the optimal strategy generated by Q-greedyUCB is completely overlapping with the optimal strategy generated by PI, ARL, and Q-learning algorithm. The global optimal scheduling strategy obtained by Algorithm~\ref{alg:1} minimizes the unconstrained optimization problem {\bf P2}.
\begin{figure}[t!]
  \centering
   \subfigure[$B = 10$, $M = 5$, $C = 4$.]{
         \includegraphics[width=7cm]{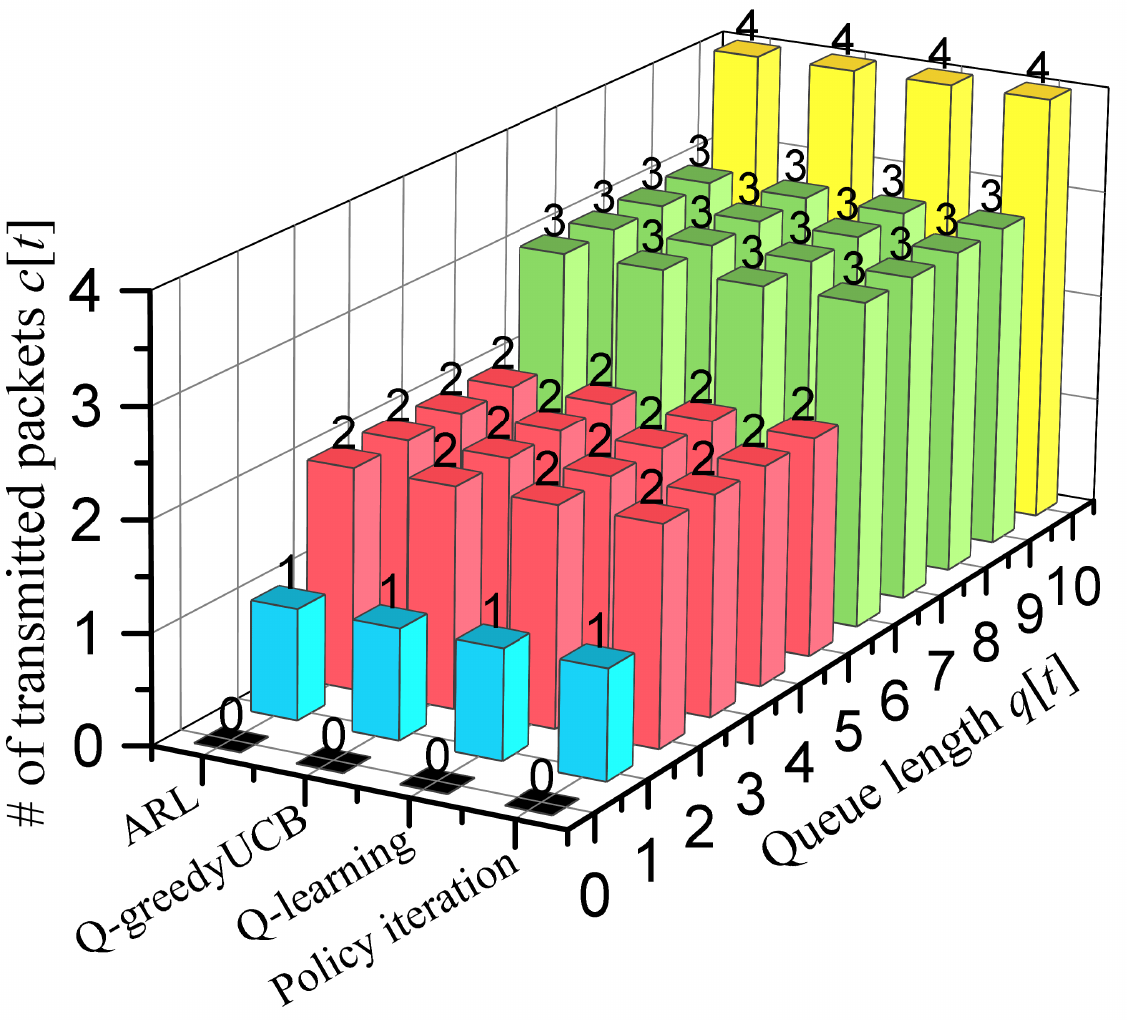} \label{fig:a1}
    }
    \subfigure[$B = 12$, $M = 5$, $C = 5$.]{
         \includegraphics[width=7cm]{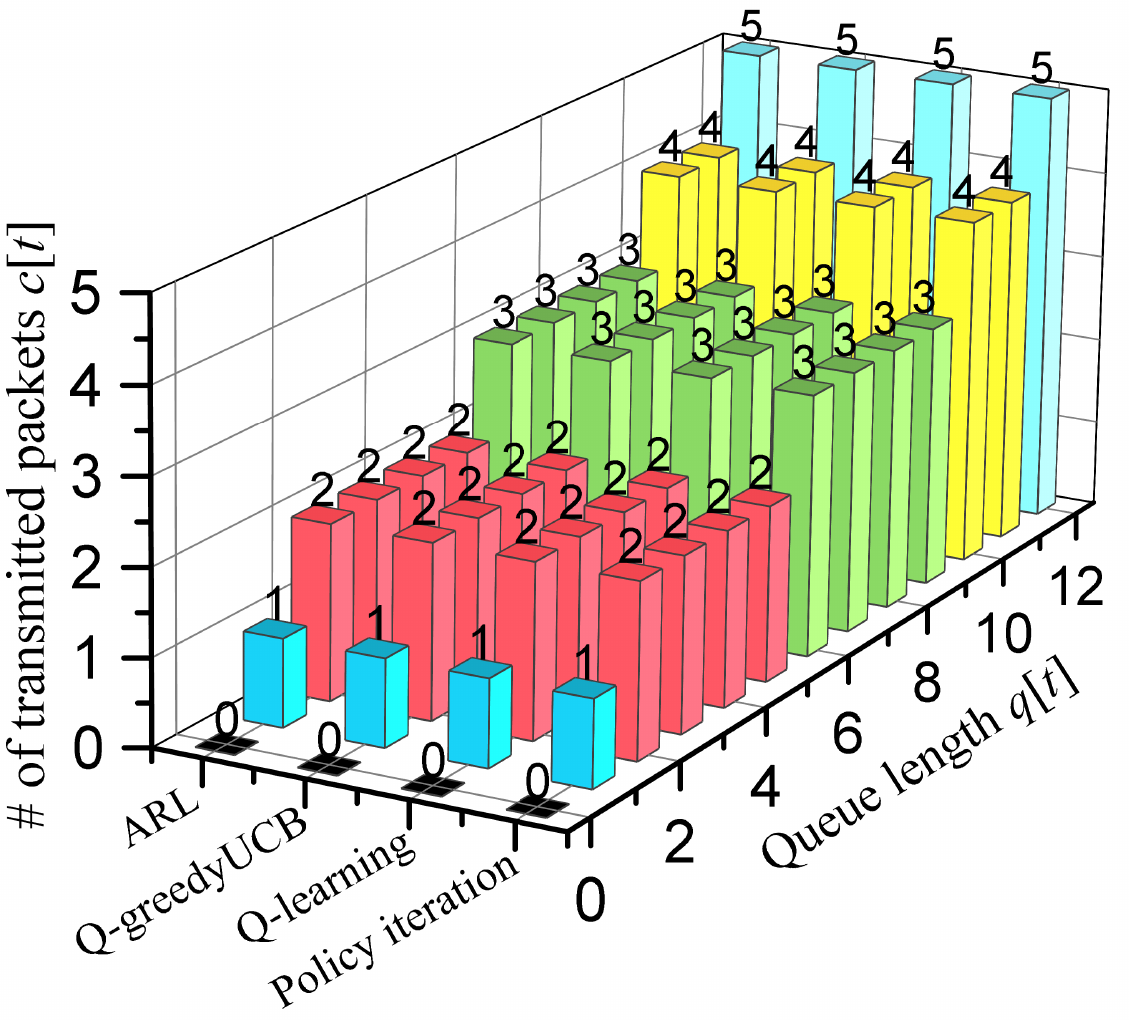}\label{fig:a2}
    }
    \caption{Deterministic optimal scheduling policies of the ARL, Q-learning, policy iteration, and the Q-greedyUCB algorithm. ($\alpha = 0.4$, $\sigma = 1$, $\varepsilon = 0.01$, $\delta = 0.01$, and $\lambda = 1$).}
    \label{fig:policy}
\end{figure}

Based on the input parameters in Fig.~\ref{fig:policy}, the average reward curve during the learning period for ARL, Q-learning, and Q-greedyUCB algorithm is shown in Fig.~\ref{fig:average_cost}. It is shown that the average reward eventually converges to -5.6 and -7.64 (i.e., optimal average reward). 
Furthermore, Q-greedyUCB generally performs better than Q-learning and ARL in terms of average reward, except in the first $400$ steps in this case, when it selects randomly among the as-yet untried actions. This is the reason why the average reward of Q-learning and ARL in the middle of the curve is better than the Q-greedyUCB algorithm, but in the long run, Q-greedyUCB is more efficient than Q-learning and ARL. 

From the trend of the curve in Fig.~\ref{fig:average_cost}, we can conclude that the average reward of Q-greedyUCB, Q-learning and ARL will eventually meet with the optimal average reward line over time\footnote{Due to the limitations of computer storage, we only show the results of $1.0 \times 10^7$ iterations.}. However, the Q-greedyUCB algorithm takes less time (about $1.0 \times 10^6$ iterations) compared with Q-learning and ARL algorithm. Since the $\varepsilon$-greedy exploration policy is adopted in Q-learning and ARL, the agent always has the opportunity to explore with a certain probability. In contrast, the UCB policy always selects the action that has the highest reward without further exploration when having enough confidence (i.e., $b_k \to 0$).

For UCB exploration policy, the performance of the algorithm is usually measured by cumulative regret. Regret refer to the difference between the total expected reward using policy $\pi$ for $t$ rounds(or iterations) and the total expected reward obtained by the agent over $t$ iterations. To better show the regrets of different algorithms, the average reward curves of different algorithms in Fig.~\ref{fig:average_cost} are shown in different graphs, as shown in Fig.~\ref{fig:regret}\footnote{Since the average reward curves of ARL and Q-learning algorithms are extremely similar, we only use the average reward curve of Q-learning.}. The area of the grey region represents the total regret in Fig.~\ref{fig:regret}. The simulation results demonstrated that the regret of Q-greedyUCB is smaller compared with Q-learning and ARL. Also, we can see that when the time step exceeds $10^6$, the regret is almost 0, which means that the agent always chooses the optimal scheduling policy. This also shows that the convergence speed of Q-greedyUCB algorithm is faster. 

\begin{figure}[t!]
  \centering
  \vspace{-0.3cm}
   \subfigure[$B = 10$, $M = 5$, $C = 4$.]{
         \includegraphics[width=7cm]{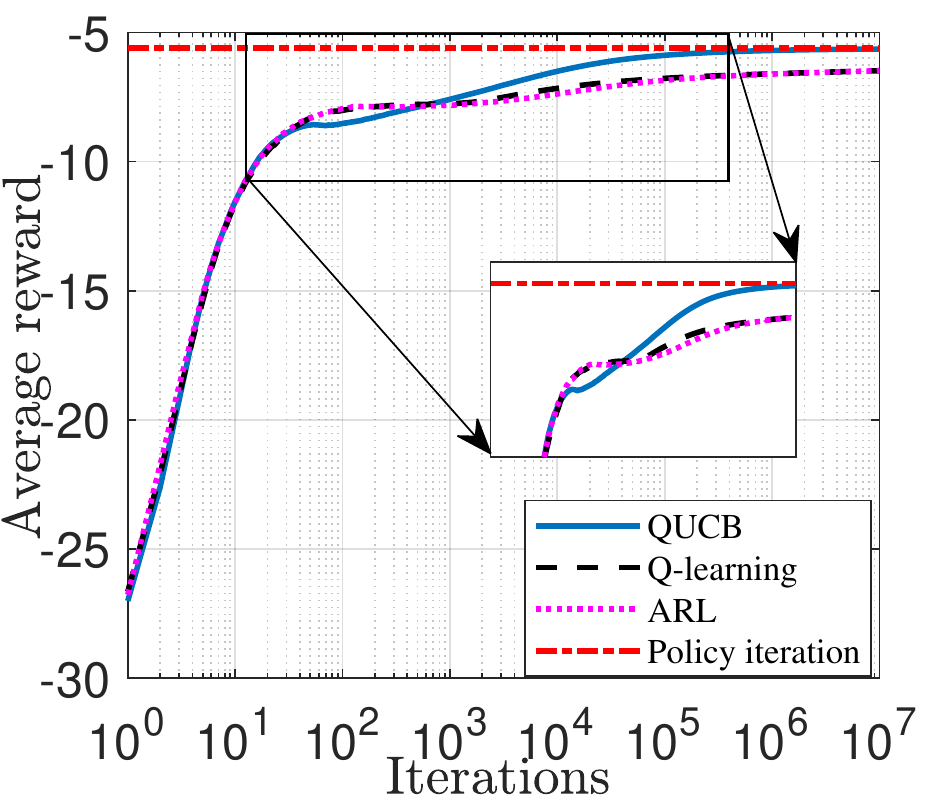} \label{fig:p1}
    }
    \subfigure[$B = 12$, $M = 5$, $C = 5$.]{
         \includegraphics[width=7cm]{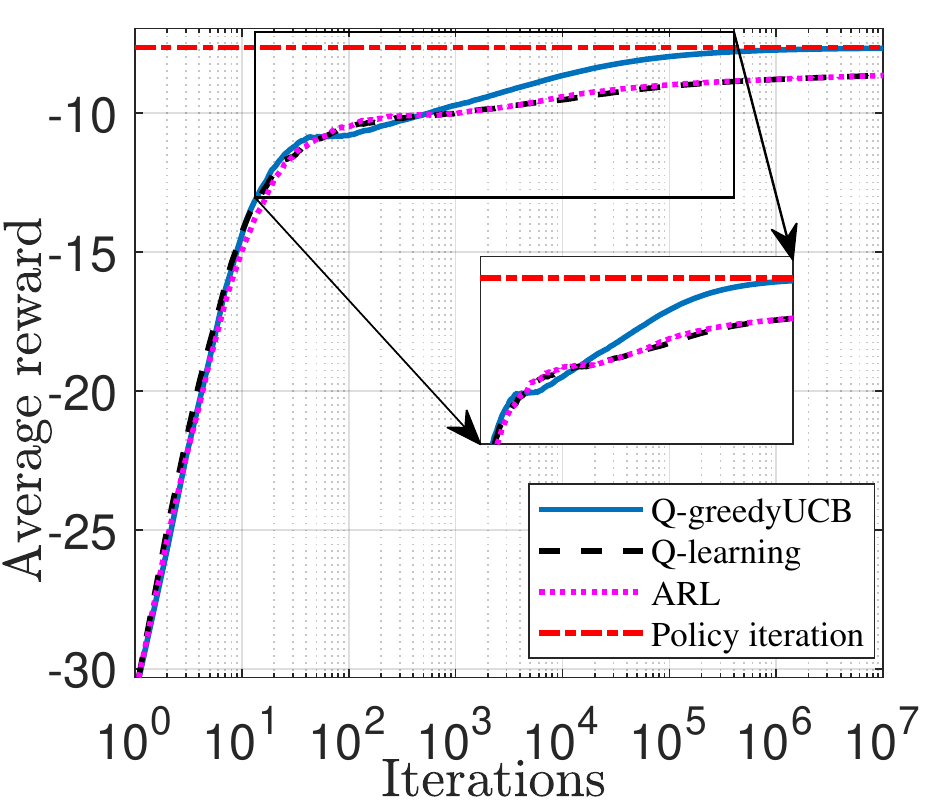}\label{fig:p2}
    }
    \caption{The average reward of Q-learning, ARL, and Q-greedyUCB over time slots and the optimal average reward from policy iteration (PI). ($\alpha = 0.4$, $\varepsilon = 0.01$, $\sigma = 1$, $\delta = 0.01$, and $\lambda = 1$).}
    \label{fig:average_cost}
\end{figure}


\begin{figure}[t!]
  \centering
  \vspace{-0.3cm}
   \subfigure[Q-greedyUCB.]{
         \includegraphics[width=7cm]{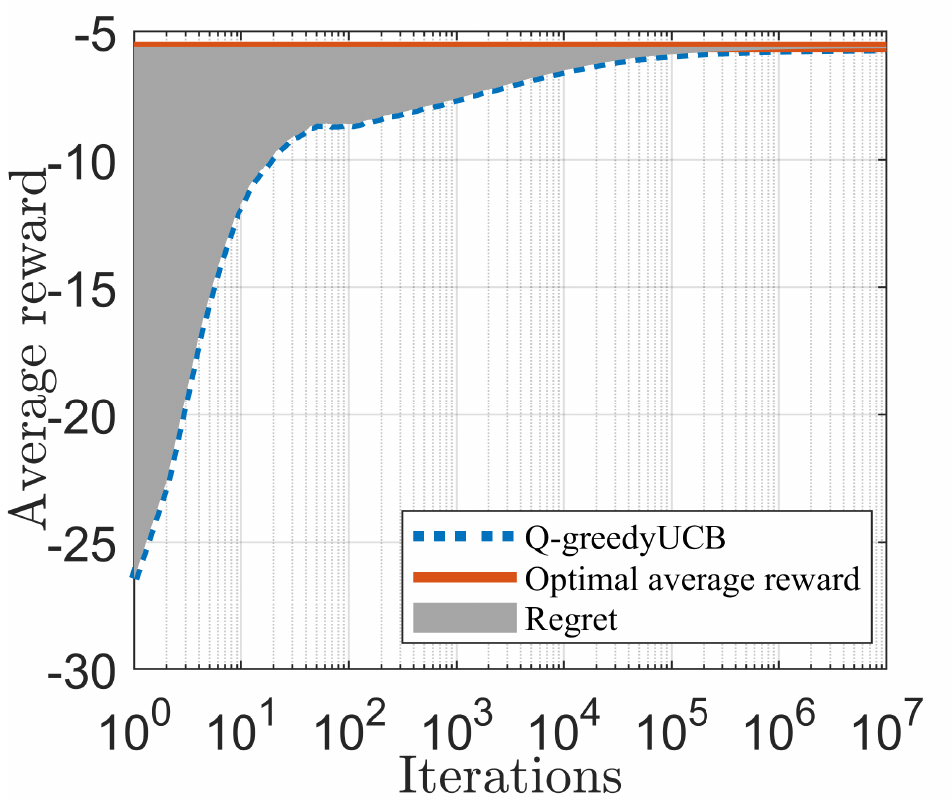} \label{fig:regret1}
    }
    \subfigure[Q-learning.]{
         \includegraphics[width=7cm]{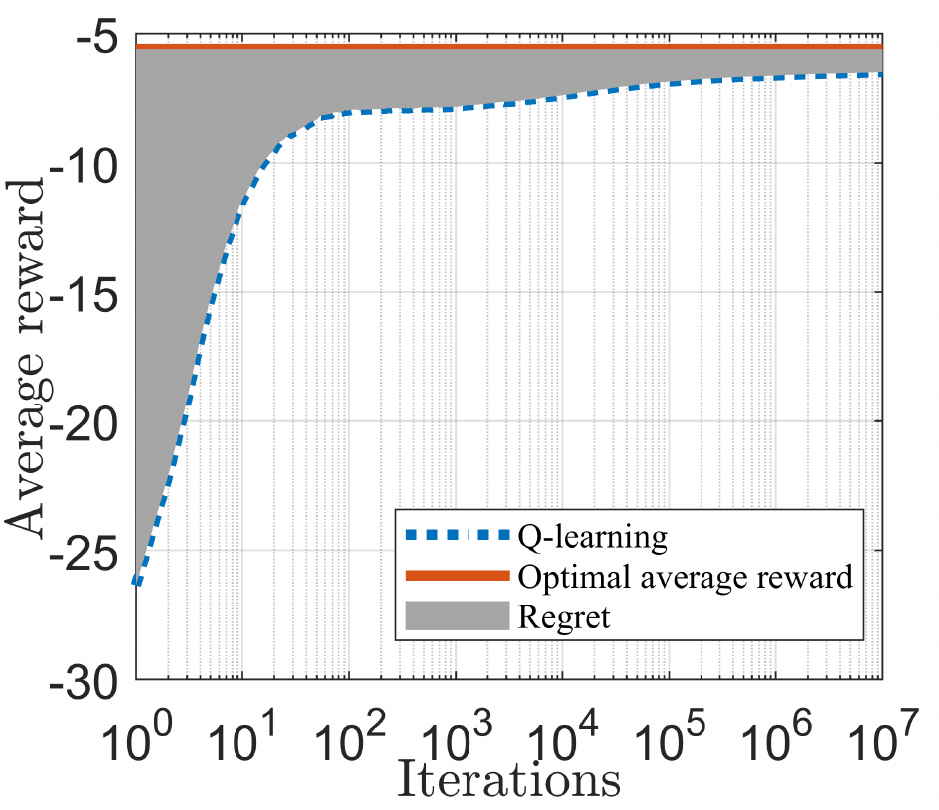}\label{fig:regret2}
    }
    \caption{The regret of Q-greedyUCB and Q-learning algorithm ($B = 10$, $M = 5$, $C = 4$, $\alpha = 0.4$, $\varepsilon = 0.01$, $\sigma = 1$, $\delta = 0.01$, and $\lambda = 1$).}
    \label{fig:regret}
\end{figure}

The optimal scheduling policy and the average reward curves generated by Q-greedyUCB are shown in Fig.~\ref{fig:arrival_date}, with $\alpha= 0.3, 0.4, 0.5, 0.6$ and $0.7$ (Other parameters remain unchanged), respectively. It is demonstrated that as $\alpha$ increases, the optimal action (service rate) and the optimal average reward value  become higher because the system workload becomes higher as the data arrival rate increases. Also, the optimal service rate becomes higher as the current queue length increases, since the latency increases when the queue stays occupied.

\begin{figure}[t!]
  \centering
   \subfigure[Optimal scheduling policy]{
         \includegraphics[width=7cm]{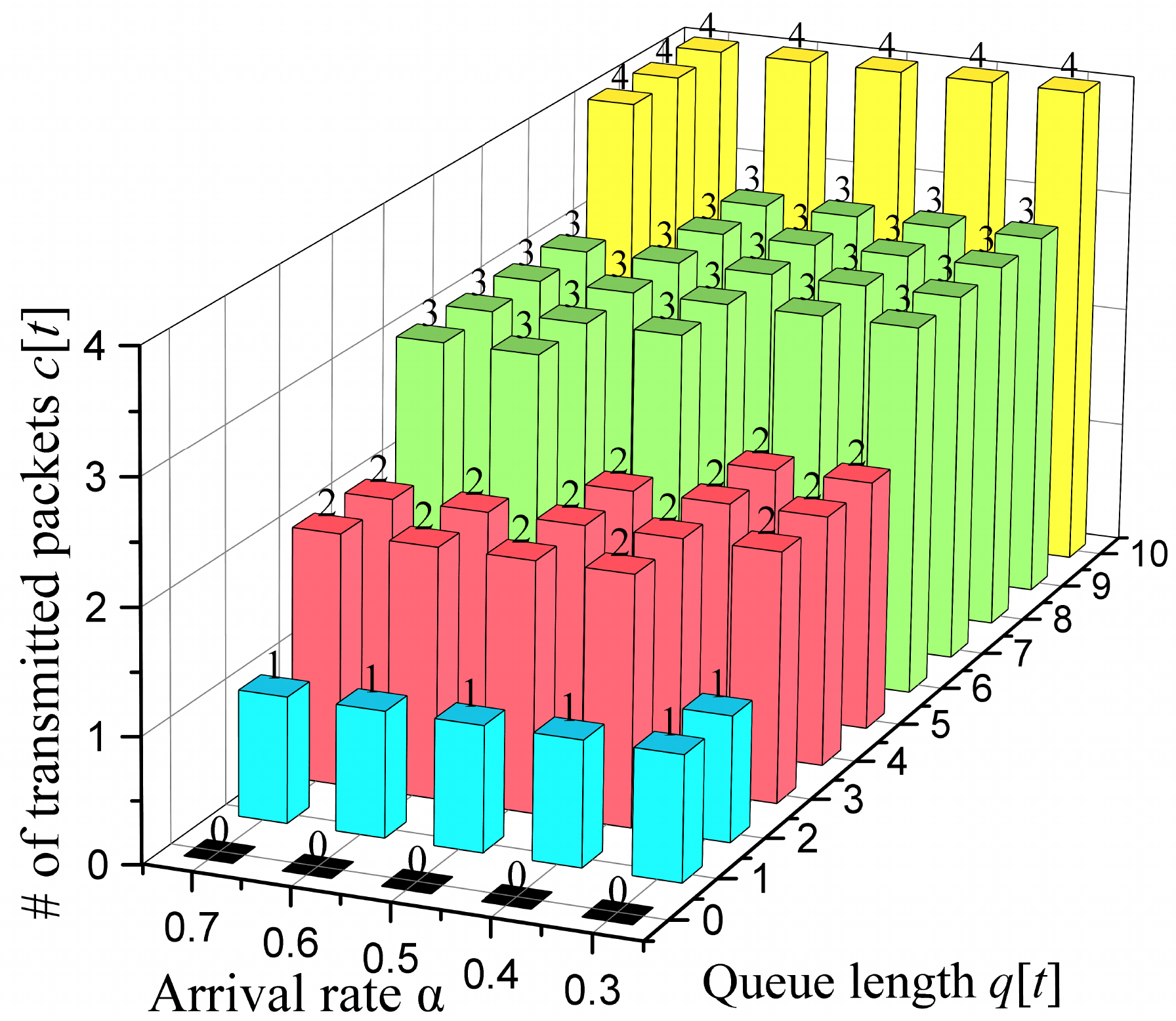} 
    }
    \subfigure[Average reward]{
         \includegraphics[width=7cm]{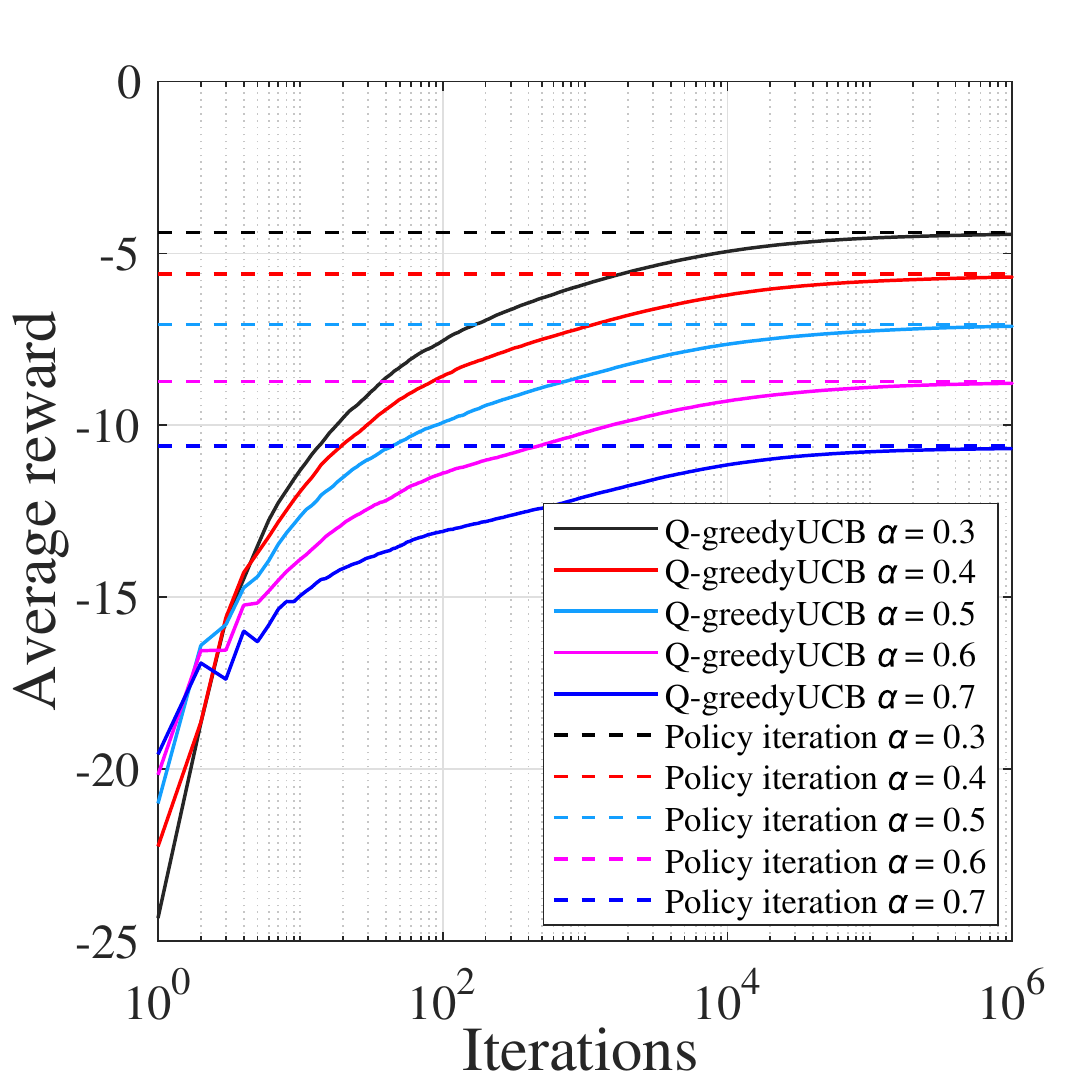}\label{fig:b1}
    }
    \caption{Deterministic optimal scheduling policies and the average reward of the Q-greedyUCB at different arrival rates ($B = 10$, $M = 5$, $C = 4$, $\sigma = 1$, $\varepsilon = 0.01$, $\delta = 0.01$, and $\lambda = 1$).}
    \label{fig:arrival_date}
\end{figure}

It should be noted that the average reward curve of Q-greedyUCB is only the average reward curve of the Q-greedyUCB algorithm during the learning process, and is not the average reward curve of the optimal policy. Additionally, the average reward at each time slot of Q-greedyUCB is the sum reward until the current time slot divided by the total length of time slots from the beginning.
Rewards incurred in the early stages do not matter since their contribution to the average reward will vanish as $T\to \infty$~\cite{bertsekas1996neuro}.

\section{Discussion and Conclusion}
\label{sec:disscussion_conclusion}
In this paper, To address the tradeoff between latency and energy consumption in communication systems, we modeled a single-queue single-server communication system. Also, in order to overcome the limitations of traditional methods (e.g., lack of flexibility) and the disadvantages of traditional RL algorithms (e.g. Q-learning), it has slow convergence speed and larger regret during training processes, we proposed a novel RL algorithm called Q-greedyUCB. In the proposed algorithm, we combined the Q-learning for \emph{average} reward algorithm with the UCB exploration strategy instead of the $\varepsilon$-greedy exploration policy adopted in the conventional Q-learning for \emph{average} reward algorithm. The simulation result shows that this method can guarantee little performance loss during the learning process compared with the conventional Q-learning algorithm and Algorithm 3 in~\cite{singh1994reinforcement}. The Lagrange multiplier method has been applied to solve the constrained optimization problem. To validate that the proposed algorithm can obtain the optimal scheduling policy, we compared it with the PI algorithm proposed in~\cite{chen2017delay}, the traditional Q-learning algorithm, and the Average-payoff RL algorithm proposed in~\cite{singh1994reinforcement}. We also mathematically prove the convergence of our Q-greedyUCB algorithm. 

Although there are important discoveries revealed by our studies, there are also limitations. The ARL, Q-learning and Q-greedyUCB algorithms are tabular solution methods. The tabular method is suitable for solving sequential decision problems with small enough state space and action space. If the dimension of the state and action space are large, this method is impractical because of the \emph{curse of dimensionality}~\cite{sutton2018reinforcement}. In the future work, we can try to use the function approximation reinforcement learning method to solve the delay-power trade-off problem in the communication system. 

RL can be used to address extremely complex problems that cannot be solved by conventional methods. In the future work, we will work on the analysis of the Q-greedyUCB in terms of the convergence rate. The impact of the step size $\gamma_t$ on convergence speed will also be studied. Furthermore, we will extend our RL method to more realistic system models (e.g., multiple channel states, general traffic arrival models, etc.).

\bibliographystyle{unsrt}  


\end{document}